\documentclass[conference,a4paper]{IEEEtran}
\pdfoutput=1

\usepackage[utf8]{inputenc}
\usepackage[margin=0.67in,top=0.545in]{geometry}
\usepackage[english]{babel}
\usepackage[T1]{fontenc}
\usepackage{epsfig}
\usepackage{amsmath, amssymb, amsbsy}
\usepackage{mathdots}
\usepackage{xspace}
\usepackage[noend]{algpseudocode}
\usepackage{algorithm}
\usepackage{algorithmicx}
\usepackage{color}
\usepackage{cite}
\usepackage{booktabs}
\usepackage{verbatim}
\usepackage[colorlinks=true,citecolor=blue,linkcolor=blue,urlcolor=blue]{hyperref}
\usepackage{url}
\usepackage{lipsum}
\usepackage{tikz}
\usepackage{enumitem}
\usetikzlibrary{arrows,matrix,positioning}
\usepackage{colortbl}
\usepackage{dblfloatfix}

\newtheorem{theorem}{Theorem}
\newtheorem{lemma}[theorem]{Lemma}
\newtheorem{corollary}[theorem]{Corollary}
\newtheorem{proposition}[theorem]{Proposition}

\newtheorem{remark}[theorem]{Remark}

\newtheorem{definition}{Definition}

\usepackage[final,tracking=true,kerning=true,spacing=true,factor=1100,stretch=10,shrink=20]{microtype}

\algrenewcommand\alglinenumber[1]{{\scriptsize#1}}
\algrenewcommand\algorithmicrequire{\textbf{Input:}}
\algrenewcommand\algorithmicensure{\textbf{Output:}}

\usepackage{varioref}        
\usepackage{fancyref}

\makeatletter
\def\mkfancyprefix#1#2{%
\expandafter\def\csname fancyref#1labelprefix\endcsname{#1}%
\begingroup\def\x{\endgroup\frefformat{plain}}%
    \expandafter\x\csname fancyref#1labelprefix\endcsname
    {\MakeLowercase{#2}\fancyrefdefaultspacing##1}%
\begingroup\def\x{\endgroup\Frefformat{plain}}%
    \expandafter\x\csname fancyref#1labelprefix\endcsname
    {#2\fancyrefdefaultspacing##1}%
\begingroup\def\x{\endgroup\frefformat{vario}}%
    \expandafter\x\csname fancyref#1labelprefix\endcsname
    {\MakeLowercase{#2}\fancyrefdefaultspacing##1##3}%
\begingroup\def\x{\endgroup\Frefformat{vario}}%
    \expandafter\x\csname fancyref#1labelprefix\endcsname
    {#2\fancyrefdefaultspacing##1##3}%
}
\makeatother
\fancyrefchangeprefix{\fancyrefeqlabelprefix}{eqn}
\mkfancyprefix{ssec}{Section}
\mkfancyprefix{tbl}{Table}
\mkfancyprefix{thm}{Theorem}
\mkfancyprefix{lem}{Lemma}
\mkfancyprefix{cor}{Corollary}
\mkfancyprefix{prop}{Proposition}
\mkfancyprefix{prob}{Problem}
\mkfancyprefix{alg}{Algorithm}
\mkfancyprefix{inv}{Invariant}
\mkfancyprefix{ex}{Example}
\mkfancyprefix{line}{Line}
\mkfancyprefix{def}{Definition}
\mkfancyprefix{itm}{Item}
\mkfancyprefix{app}{Appendix}
\mkfancyprefix{rem}{Remark}
\mkfancyprefix{property}{Property}
\newcommand{\cref}[1]{\Fref{#1}}

\definecolor{sunset}{rgb}{1,0.5,.05}
\definecolor{marine}{rgb}{0,0,.7}
\definecolor{navy}{rgb}{0,0,.5}
\definecolor{forest}{rgb}{0,.6,0}
\definecolor{brown}{rgb}{0.59, 0.29, 0.0}
\definecolor{dirtyblue}{rgb}{.4, 0.4, 0.6}

\newcommand{\alternative}[1]{{\color{brown}[/alternative: #1]}}

\renewcommand{\alternative}[1]{}

\AtBeginDocument{}

\def\ve#1{{\mathchoice{\mbox{\boldmath$\displaystyle #1$}}%
              {\mbox{\boldmath$\textstyle #1$}}%
              {\mbox{\boldmath$\scriptstyle #1$}}%
              {\mbox{\boldmath$\scriptscriptstyle #1$}}}}
\renewcommand{\vec}[1]{\ve{#1}}

\newcommand{\ZZ}{\mathbb{Z}}
\newcommand{\NN}{\ZZ_{\geq 0}}
\newcommand{\FF}{\mathbb{F}}
\newcommand{\Fq}{\mathbb{F}_q}

\newcommand{\Fsi}[1]{\mathbb{F}_{s_{#1}}}
\newcommand{\evpolys}{\mathcal{P}^{n,k}_{\tVec,\hVec,\etaVec}}

\newcommand{\Code}{\mathcal{C}}
\newcommand{\Cmult}{\Code^{n,k}(\alphaVec,\tVec,\hVec,\etaVec)}

\newcommand{\dual}{^\perp}

\newcommand{\Van}{\vec V}

\newcommand{\evword}{\mathrm{ev}}
\newcommand{\ev}[2]{\mathrm{ev}_{#2}(#1)}

\newcommand{\numTwists}{\ell}
\newcommand{\tVec}{{\ve t}}
\newcommand{\hVec}{{\ve h}}

\newcommand{\alphaVec}{{\ve \alpha}}

\newcommand{\etaVec}{{\ve \eta}}

\newcommand{\mJ}{{\ve J}}
\newcommand{\mA}{{\ve A}}
\newcommand{\mL}{{\ve L}}
\newcommand{\mG}{{\ve G}}
\newcommand{\mH}{{\ve H}}
\newcommand{\mI}{{\ve I}}

\newcommand{\transpose}{^\mathrm{T}}
\newcommand{\diag}{\mathrm{diag}}

\newcommand{\Family}{\mathcal{F}^{n,k}_{\ell}}
\newcommand{\FamilyMult}{\mathcal{\tilde F}^{n,k}_{\ell}}

\newcommand\Osoft{O^{\scriptscriptstyle \sim}\!}

\IEEEoverridecommandlockouts

\begin{document}

\title{Structural Properties of Twisted Reed--Solomon Codes with Applications to Cryptography}
\author{\IEEEauthorblockN{Peter Beelen$^{1}$, Martin Bossert$^{2}$, Sven Puchinger$^{3}$, and Johan Rosenkilde$^{1}$}\\
\IEEEauthorblockA{
$^1$Department of Applied Mathematics \& Computer Science, Technical University of Denmark, Lyngby, Denmark\\
$^2$Institute of Communications Engineering, Ulm University, Ulm, Germany\\
$^3$Institute for Communications Engineering, Technical University of Munich (TUM), Munich, Germany\\
Email: \emph{pabe@dtu.dk, martin.bossert@uni-ulm.de, sven.puchinger@tum.de, jsrn@jsrn.dk}
\thanks{This work was done while S.~Puchinger was with Ulm University.}}}

\maketitle

\begin{abstract}
We present a generalisation of Twisted Reed--Solomon codes containing a new large class of MDS codes.
We prove that the code class contains a large subfamily that is closed under duality.
Furthermore, we study the Schur squares of the new codes and show that their dimension is often large.
Using these structural properties, we single out a subfamily of the new codes which could be considered for code-based cryptography: These codes resist some existing structural attacks for Reed--Solomon-like codes, i.e.~methods for retrieving the code parameters from an obfuscated generator matrix.
\end{abstract}


\section{Introduction}

\noindent
Twisted Reed--Solomon codes \cite{beelen2017twisted} are maximum distance separable (MDS) codes\footnote{Their length $n$, dimension $k$, and minimum distance $d$ fulfil $d=n-k+1$.}
that can be efficiently decoded.
Their construction is based on Reed--Solomon (RS) codes, by adding an extra monomial (``twist'') to the low-degree evaluation polynomials and choosing the evaluation points in a suitable~way.
We present a generalisation of twisted RS codes, where instead of one additional monomial, we add $\ell$ monomials (``twists'') to the evaluation polynomials, similar to the recent extension of twisted Gabidulin codes \cite{puchinger_further_2017}.
We describe a large family of these which are MDS codes.
We study the Schur square of twisted RS code and show that, unlike for RS codes, the dimension of the Schur square is not small.
Moreover, we prove that the dual of a large class of twisted RS codes is equivalent to a twisted RS code.
Finally, we show that decoding is feasible for codes with few twists.

As a potential application, we consider the use of twisted RS codes in the McEliece cryptosystem \cite{mceliece1978}, which is a public-key cryptosystem and one of the candidates for post-quantum cryptography: a structured code $\Code$ with an efficient decoding algorithm is the ``secret key'', while an obfuscated generator matrix $\tilde \mG$ of $\Code$ is made public, together with the decoding algorithm's decoding radius $\tau$.
Encryption consists of encoding a secret message with $\tilde \mG$ and adding $\tau$ errors at random.
A ``structural attack'' on the system is to recover an efficiently decodable code that is sufficiently close to $\Code$ so that decoding the encrypted message is feasible.
Such an attack is mostly interesting if it is faster than a generic decoding algorithm, e.g.~information-set decoding~see e.g.~\cite{prange1962use,mceliece1978,peters2010information}.
The original proposal \cite{mceliece1978}, which remains unbroken, uses binary Goppa codes.
For most other proposed codes, efficient structural attacks have been found (cf. \cite{overbeck2009code}).
Attacks on RS-like codes were presented in \cite{sidelnikov1992insecurity,wieschebrink2006attack,wieschebrink2010cryptanalysis,couvreur2014distinguisher}.

We single out a family of twisted RS codes for which some of these attacks will provably not work.
For the only two other attacks that we know of, Wieschebrink's attack \cite{wieschebrink2006attack} on the dual code and Wieschebrink's squaring attack \cite{wieschebrink2010cryptanalysis}, we give compelling arguments for why they should not work, but more thorough analysis is needed.
The number of twists in our codes is $\ell= \lfloor\tfrac{1}{R}\rfloor$ for long enough codes, where $R$ is the code rate.

The new codes are over large fields with field size $q \approx n^{2^\ell}$.
This increases the size of storing the public key at a given length and dimension, i.e.~the generator matrix, but since the complexities of generic attacks strongly depend on the field size, quite short lengths are sufficient for a target security level.
We give example parameters for security levels $2^{100}$ and $2^{128}$ whose public key sizes are reduced by factors $2.7$ and $7.4$ compared to suggested parameters when using binary Goppa codes \cite{canteaut1998cryptanalysis,Barbier_2011} for the same security levels.

Notationally, matrices and vectors are generally bold-face, e.g.~$\mA$ and $\vec v$.
If $\vec v$ is a vector, then $v_i$ denotes its $i$'th element.
When adding a constant to a vector, e.g. $\vec v + 1$, we really mean $\vec v + (1,\ldots,1)$.
The diagonal matrix with diagonal entries $\vec v$ is denoted $\diag(\vec v)$.
We say that two codes are ``equivalent'' if one can be obtained from the other by permuting positions and element-wise scaling with non-zero field elements.

\section{Multi-twisted Reed--Solomon Codes}

\subsection{Definition}\label{ssec:definition}

Let $\mathcal{V} \subset \Fq[X]$ be a $\Fq$-linear subspace of polynomials.
Let $\alpha_1,\dots,\alpha_n \in \Fq$ be distinct and write $\alphaVec = [\alpha_1,\dots,\alpha_n]$.
We call $\alpha_1,\dots,\alpha_n$ the \emph{evaluation points}.
Then we define the \emph{evaluation map} of $\alphaVec$ on $\mathcal{V}$ by
\begin{align*}
\ev{\cdot}{\alphaVec} : \mathcal{V} \to \Fq^n, \quad
f \mapsto [f(\alpha_1),\dots,f(\alpha_n)].
\end{align*}
In this notation, an $[n,k]$ Reed--Solomon (RS) code with evaluation points $\vec\alpha$ is just the image of $\ev{\Fq[X]_{<k}}{\vec\alpha}$, where $\Fq[X]_{<k}$ denotes the set of polynomials in $\Fq[X]$ of degree at most $k-1$.
A generalised Reed--Solomon code (GRS) is a code which is equivalent to an RS code.

\begin{definition}\label{def:multi_twisted_polynomials}
  Given a finite field $\Fq$ and code parameters $[n,k]$, let $\numTwists \in \ZZ_{>0}$ and $\tVec,\hVec \in \ZZ_{>0}^\numTwists$ such that
  \begin{itemize}
    \item The $t_i$ are all distinct and $1 \leq t_i \leq n-k$.
    \item The $h_i$ are all distinct and $0 \leq h_i < k$.
  \end{itemize}
  Furthermore, let $\etaVec \in (\Fq \setminus \{0\})^\numTwists$.
  The set of \emph{$(\tVec, \hVec, \etaVec)$-twisted polynomials} is given as:
  \begin{equation*}
    \evpolys = \left\{ \sum_{i=0}^{k-1} f_i X^i + \sum_{j=1}^{\numTwists} \eta_j f_{h_j} X^{k-1+t_j} \ \big|\ f_i \in \Fq \right\} \ .
  \end{equation*}
  $\evpolys$ is a $k$-dimensional $\Fq$-linear subspace of $\Fq[X]$.
  We say that the space of polynomials has $\ell$ twists, each with twist $t_i$, hook $h_i$ and coefficient $\eta_i$.

  A \emph{$(\tVec, \hVec, \etaVec)$-twisted Reed--Solomon code} is given by $\alphaVec \in \Fq^n$ with distinct entries and twisted polynomials $\evpolys$ as:
  \begin{equation*}
    \Cmult = \ev{\evpolys}{\alphaVec} \subseteq \Fq^n \ .
  \end{equation*}
\end{definition}
For brevity, we will use the phrase \emph{twisted RS codes}.
Note that $\Cmult$ indeed has dimension $k$ since the evaluation map $\ev{\cdot}{\alphaVec}$ is injective on polynomials of degree $<n$ and any $f \in \evpolys$ satisfies $\deg f \le k-1+\max_i\{t_i\} < n$.

The twisted RS codes of \cite{beelen2017twisted} are the special case of the above definition when $\ell = 1$; these will be referred to as ``single-twisted RS codes''.
Specifically, the $(*)$-twisted codes introduced in \cite{beelen2017twisted} are the codes $\Code^{n,k}(\alphaVec,1,0,\eta)$, where the set of coordinates $\alphaVec$ is a subset of $\{0\} \cup G$ for a proper subgroup $G$ of $\Fq^*$ and $(-1)^k\eta^{-1} \in \Fq^*\setminus G$. The $(+)$-twisted codes from \cite{beelen2017twisted} were defined as $\Code^{n,k}(\alphaVec,1,k-1,\eta)$, where the set of coordinates $\alphaVec$ is a subset of $\{\infty\} \cup V$ for a proper subspace $V$ of $\Fq$ and $\eta^{-1} \in \Fq\setminus V$.
Note that in \cite{beelen2017twisted} we defined evaluation at $\infty$, but that in the current paper we for simplicity refrain from using $\infty$ as an evaluation point.

\subsection{A Constructive Class of MDS Twisted RS Codes}
\label{ssec:MDS_twisted_codes}

Not all twisted RS codes as defined in the previous section are MDS.
We will now describe a simple way to choose the parameters which lead to a large family of MDS codes.
This generalises a construction in \cite{beelen2017twisted} and is also inspired by a recently proposed family of twisted Gabidulin codes \cite{puchinger_further_2017} which are MRD, i.e., MDS with respect to the rank metric.

Consider some $[n,k]$ $\ell$-twisted RS code $\Cmult$ and write $H = \{ h_i \mid i = 1,\ldots,\ell \}$.
A generator matrix $\mG \in \Fq[\etaVec]^{k \times n}$ for $\Cmult$ can be obtained as the image of $\evword_{\alphaVec}$ on the monomial-like basis of $\evpolys$; that is, row $j$ of $\mG$ equals $\ev{x^j}{\alphaVec}$ if $j \notin H$, while row $h_i$ equals $\ev{x^{h_i} + \eta_i x^{k-1+t_i}}{\alphaVec}$.
Recall that $\Cmult$ is MDS if and only if any $k$ columns of $\mG$ are linearly independent.

\begin{theorem}\label{thm:MDS_sufficient_condition}
Let $s_0,\dots,s_\ell \in \NN$ such that $\Fsi{0} \subsetneq \Fsi{1} \subsetneq \dots \subsetneq \Fsi{\ell} = \Fq$ is a chain of subfields.
Let $k < n \leq s_0$ and $\alpha_1,\dots,\alpha_n \in \Fsi{0}$ be distinct, and let $\tVec$, $\hVec$, and $\etaVec$ be chosen as in \cref{ssec:definition} and such that $\eta_i \in \Fsi{i} \setminus \Fsi{i-1}$ for $i=1,\dots,\ell$.
Then $\Cmult$ is MDS.
\end{theorem}
\begin{proof}
  We will prove that any $k \times k$ sub-matrix of $\mG$ is invertible, where $\mG$ is the generator matrix defined above.
  Consider an arbitrary such sub-matrix $\hat\mG$.
  Note that $\det \hat \mG$ will equal a sum of terms of the form $\prod_{i \notin H} \alpha_{\sigma(i)}^i \prod_{j = 1}^\ell(\alpha_{\sigma(h_j)}^{h_j} + \eta_j\alpha_{\sigma(h_j)}^{k-1+t_j})$ ranging over all $k$-permutations $\sigma$.
  In particular, we can write $\det \hat \mG = \eta_\ell u_\ell + v_\ell$ where $u_\ell, v_\ell \in \Fsi{0}[\eta_1,\ldots,\eta_{\ell-1}] \subseteq \Fsi{\ell-1}$ and both of total degree at most 1 in the $\eta_1,\ldots,\eta_{\ell-1}$.
  Since $\eta_\ell \notin \Fsi{\ell-1}$ then $\det \hat \mG = 0$ only if $u_\ell = v_\ell = 0$.
  We can continue similarly and write both $u_\ell$ and $v_\ell$ as linear expressions in $\eta_{\ell-1}$ over $\Fsi{0}[\eta_1,\ldots,\eta_{\ell-2}]$ and conclude that their coefficients all needs to be identically zero for $u_\ell = v_\ell = 0$.
  Continuing with the remaining $\eta_i$, we finally conclude that $\det \hat \mG = 0$ is only possible if a set of linear expressions of the form $\eta_1 u_1 + v_1$ with $u_1,v_1 \in \Fsi{0}$ are all identically 0 -- which none of them can be.
  Thus $\hat\mG$ is invertible, and since $\hat\mG$ was chosen arbitrarily.
\end{proof}

This construction gives quite short codes: if we use $\ell$ twists, then the field size $q$ has $q \geq n^{2^\ell}$, where $n$ is the length.

\subsection{Decoding by brute-forcing the twists}\label{ssec:decoding}

A simple decoding strategy for $\Code = \Cmult$ with $\ell$ twists is to guess $\ell$ coefficients $g_1,\ldots,g_\ell \in \Fq$ and then decode $\vec r - \ev{\sum_{i=1}^\ell g_i \eta_i X^{t_i + k-1}}{\alphaVec}$ as if it is a codeword in the $[n,k]$ RS code with evaluation points $\alphaVec$.
This will succeed when $g_i = f_{h_i}$ for $i = 1,\ldots,\ell$, where $\ev{f}{\alphaVec}$ is the sent codeword and $f \in \evpolys$.
This requires $q^\ell \geq n^{\ell 2^\ell}$ rounds of RS decoding to succeed.
Some of these rounds might output purported message polynomials $\hat f_0 + \ldots + \hat f_{k-1}x^{k-1}$ for which $g_i \neq \hat f_{h_i}$: these will not correspond to close codewords in the twisted RS code and should be sifted away.

This is a $\tau$-error correcting decoder for $\Code$ when using a $\tau$-error correcting RS decoder -- even if the minimum distance of the twisted RS code is much lower than $n-k+1$.
If $\tau \leq \lfloor\frac{d-1} 2 \rfloor$, then decoding succeeds for exactly one guess of the $g_i$.
Similarly, if we use an RS list-decoder, e.g.~the Guruswami--Sudan algorithm \cite{guruswami_improved_1999}, then the total output list size will be bounded by the Johnson bound \cite{johnson_new_1962} or the stronger Cassuto--Bruck bound \cite{cassuto_combinatorial_2004}.

The complexity will be $q^\ell$ times the cost of the RS decoder.
The current best complexity for half-the-minimum distance RS decoding is $O(n \log^2 n\log\log n)$ operations in $\Fq$ \cite{justesen_complexity_1976}, while list-decoding up to the Johnson radius using the Guruswami--Sudan algorithm has a Las Vegas randomized algorithm in $O(m^4 n(\log^2(mn) + \log(q))\log\log(mn))$ operations in $\Fq$ \cite{chowdhury_faster_2015,neiger_fast_2017}, where $m$ is the ``list size'' parameter of the decoder.
In either case, the twisted RS decoder then has a complexity of $\Osoft(n^{\ell 2^\ell + 1}4^\ell)$ bit operations, ignoring $\ell \log(n)$-factors, considering $m$ a constant, and assuming $q \in O(n^{2^\ell})$.
$m$ can be considered a constant if one takes $\tau = (1-\epsilon)J_{n,k}$ for some constant $\epsilon < 1$, where $J_{n,k} = n - \sqrt{n(k-1)}$ is the Johnson radius.

Hence, decoding is feasible only for a very small number of twists, e.g., $\ell=1,2$, and finding more efficient algorithms is an open problem.

\section{Structural properties}

\subsection{Duals}
\noindent
Twisted RS codes do not generally seem to be closed under duality; however, if we choose evaluation points which form a multiplicative group, then they are:

\begin{theorem}\label{thm:dual}
Let $\alpha_1,\dots,\alpha_n$ be a multiplicative subgroup of $\Fq^*$ and let $\Cmult$ be some twisted RS code with $\alphaVec=[\alpha_1,\dots,\alpha_n]$.
Then $\Cmult\dual$ equals $\Code^{n,n-k}(\alphaVec,k-\hVec,n-k-\tVec,-\etaVec)$ up to column multipliers.
\end{theorem}
\begin{proof}
Let $\mJ$ be the matrix with 1's on the anti-diagonal and zeroes elsewhere of suitable size.
Left-multiplying by $\mJ$ reverses rows while right-multiplying reverses columns.
We denote by $\Van$ the $n \times n$ Vandermonde matrix over $\alphaVec$, i.e.
\[
\Van := [ \alpha_i^{j-1} ]_{i=1,\ldots,n, \, j=1,\ldots,k} \ .
\]
Since the entries of $\alphaVec$ form a multiplicative group, we have $\alpha_i^n = 1$ for all $i$ and by \cite{althaus1969inverse}, we obtain
\begin{align*}
\big( \Van \transpose \big)^{-1}
  &= \mJ \cdot \Van \cdot \diag(\alphaVec/n)
\end{align*}
Similar to \cref{ssec:definition}, a generator matrix of $\Cmult$ is given by $\mG = [\mI \mid \mL] \cdot \Van$, where
\begin{align*}
\mL_{ij} &= \begin{cases}
\eta_\mu, & \text{if } (i,j) = (h_\mu+1,t_\mu) \\
0, & \text{else}.
\end{cases}
\end{align*}
We claim that a parity check matrix for $\Cmult$ is:
\begin{align*}
\mH = [\mI \mid -\mJ \mL\transpose \mJ ] \cdot \Van \cdot \diag(\alphaVec/n) \ .
\end{align*}
$\mH$ has rank $n-k$ so left is to show $\mG \cdot \mH\transpose = \ve 0$.
Note that
\[ \mH = \mJ [-\mL\transpose \mid \mI] \mJ \Van \diag(\alphaVec/n) = \mJ [-\mL\transpose \mid \mI](\Van^{-1})^T \ . 
\]
Therefore,
$
\mG \cdot \mH\transpose
= [\mI \mid \mL] \begin{bmatrix}
-\mL \\ \mI
\end{bmatrix} \mJ\transpose = \ve 0
$
and the entries of $-\mJ \mL\transpose \mJ$ are of the form
\begin{align*}
\resizebox{\columnwidth}{!}{
$(-\mJ \mL\transpose \mJ)[i,j] = \begin{cases}
-\eta_\mu, & (i,j) = (n-k-t_\mu+1,k-h_\mu) \\
0, & \text{else}.
\end{cases}$
}
\end{align*}
In other words, a twist $x^{h_\mu} + \eta_\mu x^{k-1 + t_\mu}$ becomes the twist $x^{n-k-t_\mu} + (-\eta_\mu) x^{(n-k-1) + (k-h_\mu)}$ in the dual code.
\end{proof}
Note that \cref{thm:dual} implies that the dual of a $(*)$-twisted code from \cite{beelen2017twisted} is equivalent to a $(*)$-twisted code as long as $0$ is not a coordinate of $\alphaVec.$

\subsection{Schur Squares}
\label{ssec:squares}

There has been a rising interest in the \emph{Schur product} of codes both as an independent object of study, but also due to its occurrence in applications of codes, see e.g.\cite{cramer_secure_2015,randriambololona2015products} and references therein.
In particular, it is key in the structural attacks on the McEliece cryptosystem using certain Goppa codes  or other RS-like codes \cite{couvreur2014distinguisher}.

\begin{definition}
The \emph{Schur product}, or  \emph{component-wise} product, of two vectors $\ve x,\ve y \in \Fq^n$ is
$\ve x \star \ve y := [x_1 \cdot y_1,\dots,x_n \cdot y_n]$.
The Schur product of two linear codes $\Code_1, \Code_2 \subseteq \Fq^n$ is the set
\begin{align*}
  \Code_1 \star \Code_2 := \langle \ve c_1 \star \ve c_2 \, \mid \ve c_1 \in \Code_1, \ve c_2 \in \Code_2 \rangle_{\Fq} \ .
\end{align*}
We write $\Code^{2} := \Code \star \Code$ for the Schur square code.
\end{definition}

An $[n,k]$ code $\Code$ fulfills $ \dim \Code^2 \leq \min\{n, \tfrac 1 2 k (k-1)\}$.
Furthermore, if $\Code$ is MDS then $\dim \Code^2 \geq \min\{2k - 1, n\}$ \cite[p.31]{randriambololona2015products}.
GRS codes attain this lower bound with equality.

Since twisted RS codes are obtained by evaluating specific polynomials, we show below how to obtain a lower bound for the dimension of its Schur square.
For a given evaluation vector $\alphaVec$, we denote for a polynomial $f \in \Fq[X]$ by $\overline{f}$ the polynomial remainder of $f$ modulo $\prod_{i=1}^n(X-\alpha_i).$
\begin{theorem}\label{thm:square_code_dimension_lower_bound}
Let $\overline{D}=\{\deg (\overline{f\cdot g}) \mid f,g \in \evpolys\}$. Then $\dim \Cmult^2 \ge \left|\overline{D} \right|.$
\end{theorem}
\begin{proof}
First of all note that $\ev{f}{\alphaVec}=\ev{\overline{f}}{\alphaVec}$, since the evaluation map $\ev{\cdot}{\alphaVec}$ vanishes on any multiple of $\prod_{i=1}^n(X-\alpha_i).$ Further, the evaluation map $\ev{\cdot}{\alphaVec}$ is an injective homomorphism on $\Fq[X]_{<n}$, the space of polynomials of degree at most $n-1$. This implies that $$\dim \Cmult^2 = \dim \langle \overline{f \cdot g} \mid f,g \in \evpolys \rangle \ge \left|\overline{D} \right|. \IEEEQEDhere$$
\end{proof}
Especially if the evaluation vector consists of elements of a multiplicative subgroup of $\Fq^*$, the bound in the theorem is easy to compute, since in that case $\prod_{i=1}^n(X-\alpha_i)=X^n-1$.

Alternatively, the bound implies the following simpler formulation for which no remainders need be computed:
\begin{corollary}\label{cor:D}
Let $D=\{\deg (f\cdot g) \mid f,g \in \evpolys\}$. Then, $$\dim \Cmult^2 \ge \left|\{d \in D \mid d<n \}\right|.$$
\end{corollary}
\begin{proof}
This follows from $\{d \in D \mid d<n \} \subset \overline{D}.$
\end{proof}
This corollary is usually easy to apply: let $S$ be the degrees occurring in the usual monomial-like  basis of $\evpolys$, i.e.:
$
S = (\{0,\ldots,k-1\} \setminus \{ h_i \}_{i=1,\ldots,\ell }) \cup \{ t_i + k-1 \}_{i=1,\ldots,\ell} .
$
Then $D = \{ d_1 + d_2 \mid d_1, d_2 \in S \} \cap \{ 0,\ldots, n-1 \}$.

\subsection{Separation from GRS codes}
\label{ssec:separation}

\def\Couter{\mathcal{D}_{\textsc{out}}}
\def\Cinner{\mathcal{D}_{\textsc{inn}}}
Consider a twisted RS code $\Code$.
Since twisted RS codes are close kin to GRS codes, one can consider the question of finding GRS codes $\Cinner, \Couter$ such that
\[
  \Cinner \subset \Code \subset \Couter \ ,
\]
while maximising $\dim \Cinner$ and minimising $\dim \Couter$, i.e.~the ``separation'' of $\Code$ from GRS codes.
Note that $\Code \subset \Couter \iff (\Couter)^\perp \subset \Code^\perp$ and that $(\Couter)^\perp$ is a GRS code itself where we now try to maximise $\dim (\Couter)^\perp$.
When applying the codes to cryptography, we will motivate further why we consider these questions.
The following three statements separate any code $\Code$ from GRS codes from above and below based entirely on the dimension of $\Code^2$.
Since many twisted RS codes have large Schur square, this provides separation from GRS codes.

\begin{lemma}\label{lem:GRSapprox}
Let $\Cmult$ be a twisted RS code. Then $\Cmult$ contains an RS code of dimension $\min_i\{h_i\}$ and is contained in an RS code of dimension $k+\max_i{t_i}.$
\end{lemma}
\begin{proof}
  This follows from the observations that
  \begin{align*}
    \{x^j \mid 0 \le j \le \min_i\{h_i\}-1\} \subset \evpolys \subset \qquad \\
    \qquad \qquad \{x^j \mid 0 \le j \le k-1+\max_i\{t_i\}\} \ . \quad \IEEEQEDhere
  \end{align*}
\end{proof}
The description in the preceding section on the Schur square $\dim \Code^2$ of a twisted RS code $\Code$ gives a way to prove $\Code$ is non-GRS: if $\dim \Code^2 > 2k-1$, or if $\dim(\Code\dual)^2 > 2(n-k)-1$, then it must be non-GRS.
This can be strengthened to bound the dimension of the smallest GRS code which contains $\Code$ or $\Code\dual$ supplementing the observation from Lemma \ref{lem:GRSapprox}.

\begin{proposition}\label{prop:dim_Couter_k<n/2}
  Let $\Code$ be an $[n,k]$ code with $k < n/2$ and such that $\dim \Code^2 = 2k-1 + \delta$ for $\delta > 0$.
  Then if $\Couter$ is a GRS code with $\Code \subseteq \Couter$, then $\dim \Couter \geq k + \delta/2$.
\end{proposition}
\begin{proof}
  Since $\Code^2 \subseteq \Couter^2$ then $\dim(\Couter^2) \geq 2k-1 + \delta$ i.e.~$\dim(\Couter) \geq k + \delta/2$.
\end{proof}

\begin{proposition}
  \label{prop:dim_Cinner}
  Let $\Code$ be an $[n,k]$ code with $k < n/2$ and such that $\dim \Code^2 = 2k-1 + \delta$ for $\delta > 0$.
  If $\Cinner$ is a GRS code with $\Cinner \subseteq \Code$, then
    $\dim \Cinner \leq  \sqrt{\left(k-\tfrac52\right)^2-2\delta} + \tfrac 5 2$.
\end{proposition}
\begin{proof}
  Let $K = \dim \Cinner$.
  Pick a basis $\ve c_1,\ldots,\ve c_K$ of $\Cinner$ and extend this basis with $k-K$ vectors $\vec c_{K+1},\ldots,\vec c_{k}$ to form a basis for $\Code$.
  Then $\Code^2$ must be spanned by a basis of $\Cinner^2$ together with all Schur products $\ve c_i \star \ve c_j$ with $1 \leq i \leq K < j \leq k$, together with all Schur products $\ve c_i \star \ve c_j$ with $K < i, j \leq k$.
  In total
  \[
    \dim \Code^2 \leq (2K-1) + K(k-K) + \tfrac 1 2 (k-K) (k-K-1) \ .
  \]
This implies $\left(K-\tfrac52\right)^2 \leq \left(k-\tfrac52\right)^2-2\delta$.
\end{proof}

\begin{remark}
  A simpler question than separation from GRS codes is simply inequivalence to GRS codes.
  We considered this question for single-twisted RS codes in \cite{beelen2017twisted}, using different tools; unfortunately the proof of \cite[Theorem~18]{beelen2017twisted} contains a mistake and its statement is not true in general.
  However, the Schur square arguments above show that still most twisted RS codes are not GRS codes.
\end{remark}

\section{Applied to the McEliece Cryptosystem}

In this section, we present a subfamily of twisted RS codes which provably resist the known structural attacks on RS-like codes by Sidelnikov--Shestakov \cite{sidelnikov1992insecurity} and Couvreur et al.~\cite{couvreur2014distinguisher}.
We also discuss three attacks by Wieschebrink: that of \cite{wieschebrink2006attack} does not apply, and the same attack on the dual code \emph{seems} to not apply.
Thirdly, the squaring attack \cite{wieschebrink2010cryptanalysis} also \emph{seems} to not apply.

\begin{definition}
  \label{def:cryptofamily}
  Choose positive integers $k < n \leq q_0-1$ with $q_0$ a prime power and $2\sqrt n+6 < k \leq \tfrac{n}{2}-2$.
  Furthermore, choose $\ell \in \NN$ such that $\frac {n+1} {k - \sqrt n} - 2 < \ell < \min\{k+1;\ \frac{2n} k - 2;\ \sqrt n - 4 \}$.
  Let $r := \lceil \tfrac{n+1}{\ell+2} \rceil+2$, let $q = q_0^{2^\ell}$, and for $i=1,\dots,\ell$, let
  \begin{equation*}
  t_i = (i+1)(r-2) -k+2 \quad \text{and} \quad h_i = r-1+i.
  \end{equation*}
  Then the family of codes $\Family$ is the set of all codes $\Cmult$, where $\alphaVec \in (\FF_{q_0} \setminus \{0\})^n$ with distinct elements, and where $\eta_i \in \FF_{q_0^{2^i}} \setminus \FF_{q_0^{2^{i-1}}}$.
  Further, for $n \mid q_0-1$, we let $\FamilyMult \subset \Family$ be those codes which furthermore satisfy that $\alphaVec$ form a multiplicative subgroup of $\FF_{q_0}$.
\end{definition}

It is technical but easy to show that for any $n \geq 72$ and any $k$ satisfying the restrictions, then there is always a valid choice for $\ell$.
Also, the family $\Family$ is well-defined (cf.~\cref{def:multi_twisted_polynomials}) for all triples $(n,k,\ell)$ satisfying the restrictions.

By \cref{thm:MDS_sufficient_condition} all codes in $\Family$ are MDS.
It is clear that for a fixed rate $R = k/n$, the lower bound for $\ell$ in \cref{def:cryptofamily} converges to $\tfrac 1 R - 1$ for $n \to \infty$, so for large enough $n$, it suffices to choose $\ell = \lfloor\tfrac{1}{R}\rfloor$.

\begin{remark}
  If one desires a code rate greater than $\tfrac{1}{2}$, for $n \mid q_0-1$ one can use the dual of a code in $\FamilyMult$. 
  These will be twisted RS codes with $\ell$ twists, by \cref{thm:dual}, and will resist the attacks since the primal codes do.
\end{remark}

\subsection{Resistance to Schur Square Distinguishing}\label{ssec:resistance_schur_square}

Couvreur et al.\ \cite{couvreur2014distinguisher} built a structural attack on McEliece with RS codes using the fact that the square code of a $k$-dimensional RS code $\mathcal{D}$, and any shortening at up to two positions, has abnormally low dimension when $k < n/2$:
$\dim (\mathcal{D}^2) = \min\{2k-1,n\}$.
By comparison, a random linear code achieves $\dim(\Code^2) = \min\{\tfrac{1}{2} k (k+1),n\}$ with high probability \cite{cascudo_squares_2015}.
In this section we show that the codes of $\Family$ are impervious to this attack.

\begin{theorem}\label{thm:distinguisher_resistant_construction}
  Let $n,k,\ell$ be as in \cref{def:cryptofamily}, and let $\Code \in \Family$.
  Then $\dim \Code^2 = n$ and $\dim((\Code^\perp)^2) = n$.
\end{theorem}
\begin{proof}
  Let $r, \vec t, \vec h$ be as in \cref{def:cryptofamily}.
  We will use \cref{cor:D}, so let $D = \{ \deg(f \cdot g) \mid f,g \in \evpolys \}$.
  Since $x^0,\ldots,x^{r-1} \in \evpolys$ then $0,\ldots,2r-2 \in D$.
  For each $i$, there is a polynomial of degree $t_i+k-1$ in $\evpolys$, so also $(i+1)(r-2)+1,(i+1)(r-2)+2,\ldots,(i+1)(r-2)+r \in D$.
  Summing up we have $0,1,\ldots,(\ell+2)(r-2)+2 \in D$.
  By definition of $r$, we have $(\ell+2)(r-2)+2 \geq n-1$, so \cref{cor:D} implies the claim for $\Code$.
  We have $\dim (\Code^\perp)^2 = n$ since $\Code^\perp$ is MDS and has dimension $>n/2$ \cite[p.31]{randriambololona2015products}.
\end{proof}

By a simple, but more technical, argument it can be shown that also shortenings of the codes at up to two positions have maximal Schur square dimension. The proof uses that the shortened code has evaluation polynomials of degrees $2,\dots,r$ and $t_i+k-1$, so its Schur square has evaluation polynomials of degrees $2,\dots,n-1$ (i.e., $n-2$ consecutive degrees), which implies the claim since $0$ is not included in the evaluation points. This proves the resistance against the attack in \cite[Section~6]{couvreur2014distinguisher}.

\begin{remark}
It is not hard to come up with more twisted RS codes of Schur square dimension $n$.
The crucial ingredient of the proof of \cref{thm:distinguisher_resistant_construction} is that $x^0,\ldots,x^{r-1} \in \evpolys$.
As long as all hooks $h_i$ are at least $r$, then the argument will work.
The reason for setting $h_1$ as small as possible is discussed in the following section, but the remaining $h_i$ could be chosen as any subset of $\{r+1, \ldots, k-1 \}$.
Further, whenever $(\ell + 2)r$ is much bigger than $n-1$, the choice of $t_i$ can be perturbed in many ways.
\end{remark}

\subsection{Resistance to Sidelnikov--Shestakov \& Wieschebrink attack}

The Sidelnikov--Shestakov (SiS) attack  \cite{sidelnikov1992insecurity} completely breaks McEliece with GRS codes; 
we use the description from \cite{wieschebrink2010cryptanalysis}.
Let $[\mI \mid \mA]$ be the systematic generator matrix of a GRS code with evaluation points $\alphaVec \in \Fq^n$ and column multipliers $\vec v \in (\Fq^\star)^n$.
Using the automorphisms of a GRS code, the attacker can assume $\alpha_1 = 0$ and $\alpha_1 = 1$.
Row $i$ of the matrix is the evaluation of the polynomial
$f_i = c_i \textstyle\prod_{\mu=1, \mu \neq i}^{k} (x-\alpha_\mu)$,
where $c_i$ is a suitably chosen constant in $\Fq$.
It is well known \cite{roth1985generator,roth1989mds} that the matrix $\mA$ is a Cauchy matrix exactly when the code is GRS.
This means that for $d_j=\prod_{\nu=1}^k(\alpha_j-\alpha_\nu)$ for $j=k+1,\dots,n$, we have the simple expression
\begin{equation}
A_{i\,j}=\tfrac{c_i \cdot d_j}{\alpha_{j}-\alpha_i}.\label{eq:Cauchy_RS}
\end{equation}
Due to this compact formula, after guessing the values $c_1, c_2 \in \Fq$, the remaining secret values $\alpha_3,\dots,\alpha_n$ and $v_1,\dots,v_n$ can be computed by simple linear equations.

Since \cref{thm:distinguisher_resistant_construction} shows that the codes of $\Family$ have large Schur square, then none of these codes is GRS.
Therefore, a systematic generator matrix for these codes has $[ \mI \mid \mA ]$ where $\mA$ is \emph{not} Cauchy.
Since the SiS attack is rested on the simple Cauchy structure, it will not work for these codes.

Wieschebrink \cite{wieschebrink2006attack} generalised the SiS to apply for sub-codes of RS codes of small codimension.
Assume that we have given a $k \times n$ generator matrix $[\mI \mid \mA]$ of a subcode of an $(n,k+\delta)$ RS code.
Then, the $i$-th row of the matrix $\mA$ is obtained from a polynomial
$f_i = c_i(x) \textstyle\prod_{\mu=1, \mu \neq i}^{k} (x-\alpha_\mu)$,
where $c_i(x) \in \Fq[X]$ is a polynomial of degree $\deg c_i(x) \leq \delta$.
One now needs to guess the polynomials $c_1(x)$ and $c_2(x)$ instead of only scalars, which will succeed after at most $\approx q^{2\delta}$ iterations.
For large enough values of $\delta$, this is infeasible.

For $\Code \in \Family$, since $\dim \Code^2 = n$, then \cref{prop:dim_Couter_k<n/2} implies that any GRS code $\Couter$ with $\Code \subset \Couter$ must have $\dim \Couter \geq k + (n-2k+1)/2$.
This very quickly makes Wieschebrink's attack infeasible for $\Code$.

Wieschebrink's attack can also be carried out on the dual code, so that one seeks a GRS code $\Cinner^\perp$ such that $\Code^\perp \subset \Cinner^\perp$, i.e.~$\Cinner \subset \Code$.
If $\dim \Cinner$ is close to $\dim \Code$, this is a bad sign: it may reveal structural information on $\Code$ or may drastically reduce the cost of brute-force decoding the encrypted message.
We are at this point unable to categorically refute this attack.
\cref{prop:dim_Cinner} provides only a small separation between $\Cinner$ and $\Code$, and it is not hard to realise that this separation can never be greater than the number of twists $\ell$, which we wish to keep fairly small for decoding purposes.
As per \cref{lem:GRSapprox}, the obvious GRS code contained in $\Code$ is the one with the same evaluation points $\alphaVec$ and dimension $\min\{h_1\} = r$.
If $\Code$ contains a larger GRS code, it seems it would have to be a surprising, or spurious, code with other evaluation points and column multipliers.

If we choose $n \mid q_0-1$ and $\Code \in \FamilyMult$, then $\Code^\perp$ is itself a twisted RS code by \cref{ssec:decoding}.
Now the obvious GRS code $\Cinner$ such that $\Code^\perp \subset \Cinner^\perp$ is the one having the same evaluation points and dimension $\max\{\tilde t_i + n - k\} = n-r$, where $\tilde \tVec = k - \hVec$ is the twist vector of $\Code^\perp$.
Perhaps this strengthens the belief that it is not likely that $\Code$ contains another, surprisingly large GRS code.

Summarising, assuming that $\Code \in \Family$ contains no surprisingly large GRS code, then Wieschebrink's attack on $\Code$ or $\Code^\perp$ will have work factor roughly $q^{2\delta}$ if
\[
  k \in [r+\delta,\tfrac{n+1}{2}-\delta] \approx [\tfrac{n}{\ell+1}-\delta,\tfrac{n}{2}-\delta] \ .
\]

\subsection{Wieschebrink's Squaring Attack}
Wieschebrink gave yet another attack in \cite{wieschebrink2010cryptanalysis} for McEliece with subcodes of RS codes:
given a scrambled generator matrix for the $[n,k]$ code $\Code$ which is a subcode of an unknown GRS code $\Couter$, the attack finds the parameters of $\Couter$.
The idea is that $\Code^2$ and $\Couter^2$ agree with high probability, if $\Code$ is a random subcode of $\Couter$ with not too small co-dimension.
If $\dim \Couter > n/2$ then one can first shorten $\Code$ and $\Couter$ at enough random positions and then consider their square.
Applied to twisted RS codes, if $\Code \in \Family$ then we saw in the previous subsection that $\dim \Couter > n/2$, thus the squaring attack should first shorten $\Code$ at some random positions.
If we can show that $(\Code^*)^2$ is never, or rarely, a GRS code, where $\Code^*$ is some shortening of $\Code$, then $\Code$ is resistant to (straightforward version of) the squaring attack.

At the time of writing, we don't have a such a theorem.
However, we believe the following approach is promising:
assume that we shorten at the first $s$ positions and write $\alphaVec = (\alphaVec_S | \alphaVec^*)$ with $\alphaVec_S$ having length $s$.
Then $C^* = \ev{\mathcal{P}}{\alphaVec^*}$ where $\mathcal{P} = \{ f \in \evpolys \mid f(\alpha) = 0 \textrm{ for } \alpha \in  \alphaVec_S \}$.
Hence $\dim (C^*)^2$ can be characterised using the tools of \cref{ssec:squares} by looking at $\bar D := \{ \deg\big( (f \cdot g) \mod M  \big) \mid f, g \in \mathcal P \}$, where $M = \prod_{i=s+1}^{n}(X - \alpha_i)$.
Since the set $\{ \deg f \mid f \in \mathcal P \}$ will likely have gaps as well as a smallest element of at least $s$, it seems unlikely that $\bar D$ will contain only $2(k-s)$ elements.
Then $(C^*)^2$ is not a GRS code.

\subsection{Example Code Parameters}
\label{ssec:example}

We consider the following example parameters:
Let $(n,k) = (255,117)$, $\ell = 1$, $q_0 = 2^8$ and pick a code over $\Fq, q = 2^{16}$, with these parameters as in~\cref{def:cryptofamily}.
The number of such codes is larger than $n!(q-\sqrt{q}) \approx 2^{2038} \cdot 2^{16}$.
The actual number of inequivalent codes is presumably much smaller, but seems to suffice for avoiding exhaustive search attack.
Under the assumptions of the previous subsections, the work factor of Wieschebrink's attack is $2^{336}$.

As for generic decoding attacks, we consider classical information set decoding \cite{prange1962use,mceliece1978}\footnote{%
  The improvements for non-binary information-set decoding considered in \cite{peters2010information} seem to be not more efficient since, unlike the original algorithm \cite{mceliece1978}, the cost of an iteration strongly depends on the
  field size, which is large for our codes: $q=2^{16}$. However, this should be more carefully studied e.g.~by estimating the number of iterations by a Markov chain simulation as in \cite{peters2010information}.
}
with work factor $W_\mathrm{I} = \binom{n}{k}/\binom{n-\tau}{k} k^3 \log_2(q)^2$, where $\tau$ is the number of errors inserted by the sender.
Since the example codes can correct up to $\tau_\mathrm{unique} = 69$ errors uniquely and $\tau_\mathrm{list} = 83$ using a list decoder, we obtain the following work factors for unique and list decoding, respectively:
$W_\mathrm{I,unique} \geq 2^{105}$, $W_\mathrm{I,list} \geq 2^{126}$.
In both cases, the key size is
\begin{equation*}
K_\mathrm{sys} = k(n-k) \log_2(q)/8192 = 31.5 \, \mathrm{KB},
\end{equation*}
using a systematic generator matrix for the public key.

The security level of the two variants above is $\geq 2^{100}$ resp.\ $\approx 2^{128}$.
We compare the key size to parameter proposals for binary Goppa codes of similar security levels:
\begin{description}
\item[$2^{100}$ \cite{canteaut1998cryptanalysis}:] $(n,k) = (2048,1608)$, $\tau= 40$,
$K = 86.4 \, \mathrm{KB}$.
\item[$2^{128}$ \cite{Barbier_2011}:] $(n,k) = (3262,2482)$, $\tau= 66$,
$K = 236.3 \, \mathrm{KB}$.
\end{description}
In this example, twisted RS codes reduce the key size by a factor $2.7$ and $7.4$ at security levels $\approx 2^{100}$ resp.\ $\approx 2^{128}$.

\section{Conclusion}

We have presented a natural generalisation of twisted RS codes, pointed out a large subfamily of MDS codes, and presented a decoder that is feasible for $\ell=1,2$ twists.
Those codes whose evaluation points form a multiplicative group are closed under duality with explicit duals.
We also showed that the Schur square of a twisted RS code is usually large, which shows that they are in a sense ``far'' from GRS codes, cf.~\cref{lem:GRSapprox}.
Furthermore, we identified a subfamily of twisted RS codes that resist some known structural attacks on the McEliece cryptosystem: Sidenlikov--Shestakov \cite{sidelnikov1992insecurity}, Wieschebrink \cite{wieschebrink2006attack} and Schur square-distinguishing \cite{couvreur2014distinguisher}.
For Wieschebrink's attack on the dual code and Wieschebrink's squaring attack \cite{wieschebrink2010cryptanalysis} we could say only that the attack does not \emph{seem} to apply.
We gave example parameters that achieve lower key sizes than the original McEliece cryptosystem for the same security level.

Whether (some sub-families of) twisted RS codes are suitable for the McEliece cryptosystem or not remains to be seen, but we believe that our analysis motivates looking closer at this question: this includes studying Wieschebrink's attack on the dual code and Wieschebrink's squaring attack, as well as seeking completely new attacks utilizing the particular structure of twisted RS codes.

\bibliographystyle{IEEEtran}
\bibliography{main}

\end{document}